\documentclass[10pt,envcountsame]{llncs}

\usepackage[utf8]{inputenc}
\usepackage{amsmath}
\usepackage{amssymb}
\usepackage{rotating}
\usepackage{enumerate}
\usepackage{xspace}
\usepackage{amsfonts}
\usepackage{mathrsfs}
\usepackage[all,2cell]{xy}
\CompileMatrices
\UseAllTwocells
\SilentMatrices
\usepackage{cite}
\usepackage{float}
\usepackage{fancybox}
\usepackage{a4wide}
\usepackage{hyperref}
\usepackage{stmaryrd}
\usepackage{caption}
\usepackage{subcaption}

\newcommand{\Defeq}
 {\stackrel{\mathrm{def}}{=}}

\newcommand{\xRightarrow}[1]{\stackrel{#1}{\Rightarrow}}

\newcommand{\stran}{\raise1pt\hbox{$\centerdot$}}

\newcommand{\rring}[1]{\ensuremath{\mathbb{#1}}}

\newcommand{\N}{\rring{N}}

\newcommand{\ladj}[2]{\ar@/^/[#1]^-{#2} \ar@{}[#1]|-%

{\ifthenelse{\equal{#1}{r}}{\bot}{%

{\ifthenelse{\equal{#1}{rr}}{\bot}{%

{\ifthenelse{\equal{#1}{l}}{\top}{%

{\ifthenelse{\equal{#1}{u}}{\dashv}{%

{\vdash}}}}}}}}}}

\newcommand{\radj}[2]{\ar@/^/[#1]^-{#2}}

\newcommand{\radjff}[2]{\ar@{_{(}->}[#1]^{#2}}

\newcommand{\pullbacktop}[4]{%

{#1} \ar@/_/[ddr]_{#4} \ar@/^/[drr]^{#2}%

\ar@{.>}[dr]|-{#3} \\}

\newcounter{ncomm}

\newcommand{\ltsred}[1]
{ \setbox0=\hbox{$\ {}^{#1}\ $}
  \setbox1=\hbox{$\longrightarrow$}
  \loop\setbox1=\hbox{$-$\kern-0.3em\unhbox1}\ifdim\wd1<\wd0\repeat
  \hbox{$\ \ \mathop{\box1}\limits^{#1}\ \ $}
}

\newcommand{\arx}[2]{\!\xymatrix@=15pt{\ar[r]^{{#1}}_{{#2}}&}\!}

\newlength{\mylength}

{\setlength{\fboxsep}{15pt} 
\setlength{\mylength}{\linewidth}%
\addtolength{\mylength}{-2\fboxsep}%
\addtolength{\mylength}{-2\fboxrule}%
\Sbox 
\minipage{\mylength}%
\setlength{\abovedisplayskip}{0pt}%
\setlength{\belowdisplayskip}{0pt}%
}%
{\endminipage\endSbox 
\[\fbox{\TheSbox}\]}

\newcommand{\bnfSep}{\ |\ }
\newcommand{\bnfEq}{\ ::=\ }
\newcommand{\comp}{\mathrel{;}}
\newcommand{\dtrans}[2]{\lower.2em\hbox{$\xrightarrow{#1/#2}$}}
\newcommand{\dTrans}[2]{\lower.2em\hbox{$\xRightarrow{#1/#2}$}}
\newcommand{\epstrans}[2]{\lower.2em\hbox{$\xrightarrow{\epsilon_{#1,#2}}$}}

\newcommand{\pre}[1]{{^\circ{#1}}}
\newcommand{\post}[1]{{#1^\circ}}
\newcommand{\preandpost}[1]{{^\circ{#1^\circ}}}
\newcommand{\source}[1]{{^\bullet{#1}}}
\newcommand{\target}[1]{{#1^\bullet}}
\newcommand{\minsynch}[2]{sync_{min}(#1,#2)}
\newcommand{\marking}[2]{{[#1]}_{#2}}
\newcommand{\NFA}[3]{\mathsf{NFA}(#1,\,#2,\,#3)}
\newcommand{\semanticsOf}[1]{\llbracket{#1}\rrbracket}
\newcommand{\epsmin}[1]{{\epsilon\mathsf{min}(#1)}}
\newcommand{\epsminpr}[1]{{\mathsf{trans}(#1)}}
\newcommand{\epsclose}[1]{{\epsilon\mathsf{cl}(#1)}}
\newcommand{\mrk}[2]{\mathsf{mrk}(#1)_{#2}}
\newcommand{\figref}[1]{Fig.~\ref{#1}}

\title{Reachability via Compositionality in Petri nets}
\author{Pawe{\l} Soboci{\'n}ski \and Owen Stephens}
\institute{ECS, University of Southampton, UK}
\begin{document}
\maketitle
\begin{abstract}
We introduce a novel technique for checking reachability in Petri nets that
relies on a recently introduced compositional algebra of nets. We prove that
the technique is correct, and discuss our implementation. We report promising
experimental results on some well-known examples.
\end{abstract}

\section*{Introduction}
\label{sec:intro}

We introduce a novel technique for checking reachability in 1-bounded Petri nets. Our approach relies on a structural decomposition of
nets, using the algebra of \emph{nets with boundaries} developed
in~\cite{Soboci'nski2010,Bruni2011,Bruni2012} and the algebra of labelled transition systems (LTS) originally developed in~\cite{Katis1997a}. After explaining the intuitions and some motivating examples, we prove the technique correct, discuss our implementation and report on experimental results.

Many asynchronous systems are regular in their structure, in the sense that
they can be considered as a suitable composition of several identical,
communicating components. In many such systems, the communication between
individual components can be characterised using relatively small (w.r.t. the
size of the global state space) amounts of information, and as a consequence,
the reachability of a particular global state can be checked locally. The
algebra of nets with boundaries allows us to capture precisely how separate
“component nets” communicate with each other.

\begin{figure}
\[
\includegraphics[height=2.5cm]{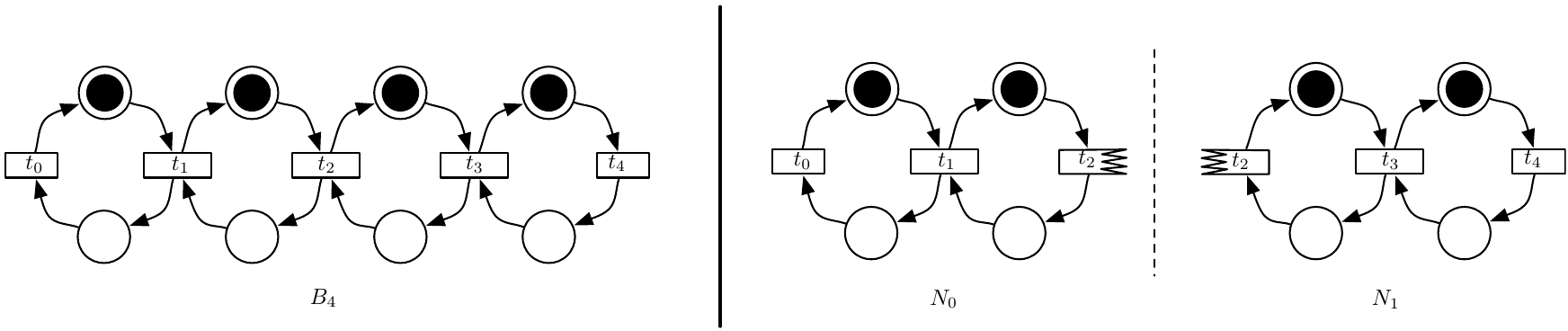}
\]
\caption{The net $B_4$ and a “cut” along its transition $t_2$.\label{fig:bookshelf}}
\end{figure}

To illustrate the ideas that underlie our approach we introduce the simple,
well-known\footnote{For example, see~\cite[Fig.\ 6]{Esparza2002}.}
bounded buffer net, $B_n$, illustrated in the left part of
\figref{fig:bookshelf}. We wish to check whether the “opposite” marking is
reachable—that is, the places in the lower row are to be marked and the places
in the upper row are to be unmarked. Taking a global view, a simple calculation
confirms that the length of the firing sequence necessary to reach the desired
marking is quadratic in $n$ (see~\figref{fig:decompositionTimes}). We will,
instead, check for reachability locally, component-wise, so imagine that the
net is “split” into two nets $N_0$ and $N_1$, sharing the transition $t_2$, as
in the right part of \figref{fig:bookshelf}.

\begin{remark}\label{rmk:policies}
Observe (1) that $N_0$ and $N_1$ can proceed independently to reach the desired
local marking, only “synchronising” on $t_2$ and (2) the “synchronisation
policy” is quite simple to describe. Indeed, $N_1$ can fire its local copy of
$t_2$ an arbitrary number (including 0) of times during a successful
computation; $N_0$ can reach its desired marking after firing its copy of $t_2$
at least twice, after which $t_2$ can be fired an arbitrary additional number
of times. These two “policies” are clearly compatible, meaning that the entire
net can reach its global desired marking.
\end{remark}

To make the above intuitions precise, we recall the algebra of nets introduced
in~\cite{Soboci'nski2010}. We will use a non-standard graphical representation
of nets, more suited for illustrating the operations of the algebra: $B_4$ is
rendered with the alternative graphical notation in the left-most diagram of
\figref{fig:bookshelfalternative}. Transitions are represented using
undirected links and each link can be connected to an arbitrary number of
ports. Each place has two ports: one for incoming transitions, illustrated with
a triangle pointing into the place, and one for outgoing transitions,
illustrated with a triangle pointing out of the place.  Thus the pre-set of a
transition is the set of places to which it is connected via their outgoing
port, and its post-set is the set of places to which it is connected via their
incoming port. Transitions can also be connected to \emph{boundary ports},
which serve as an interface between nets with boundaries. The net $B_4$ can be expressed as the composition $\top \comp b_1\comp b_1\comp
b_1\comp b_1 \comp \bot$; the individual components $\top$, $b_1$ and $\bot$
are illustrated in \figref{fig:bookshelfalternative}.
The operation ‘$\comp$’ that composes two nets along a compatible, common boundary is defined formally in \S\ref{sec:composition}.
\begin{figure}
\[
\includegraphics[height=2cm]{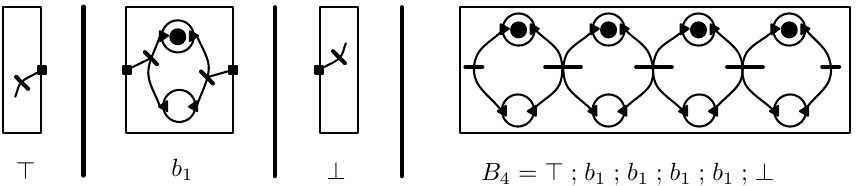}
\]
\caption{Obtaining $B_4$ as a composition of nets $\top$, $b_1$ and $\bot$.
\label{fig:bookshelfalternative}}
\end{figure}

Each component net with boundaries, together with its initial marking and
desired local marking, can be translated to a non-deterministic finite
automaton (NFA), with states being the reachable markings, and transitions the
boundary interactions observed when net transitions fire. The initial state is
the initial marking and the final state is the desired marking. We
illustrate\footnote{All illustrations of automata were generated with
GraphViz~(\url{http://www.graphviz.org}). For space-efficiency, transitions are
annotated with sets: $\{x,\,y\}$, representing the existence of two
transitions, labelled respectively $x$ and $y$. We use $*$ in the labels as
shorthand for any choice of $0$ and $1$.} this translation in
\figref{fig:piecestonfas}. For example, in the translation of $b_1$, state $0$
corresponds to the initial marking and state $1$ to the desired complementary
marking. The labels of transitions are, in general, pairs of binary strings
$\alpha$ and $\beta$, written $\alpha / \beta$, representing interaction on the
left ($\alpha$) and the right ($\beta$) boundaries. The concept of “interaction
on a boundary” is important and we will explain it further below. To guarantee
compositionality, we must use an underlying \emph{step} firing semantics of
nets, i.e. a transition in the NFA witnesses the firing of a (possibly empty)
set of independent transitions within the component net. Returning to the
translation of $b_1$: the $0/0$ labelled NFA-transitions in state $0$ and $1$
witness the possibility of no behaviour (i.e. the empty set of net-transitions
firing) with the $0/0$ label signifying that no net-transitions connected to
either boundary were fired. The NFA-transition $0\dtrans{0}{1}1$ witnesses that
the right hand side net-transition has fired and produced the desired marking.
The fact that the fired transition is connected to the port on the right
boundary is recorded by $1$ in the transition label. The remaining
NFA-transition is symmetric.

\begin{figure}
\[
\includegraphics[height=2.5cm]{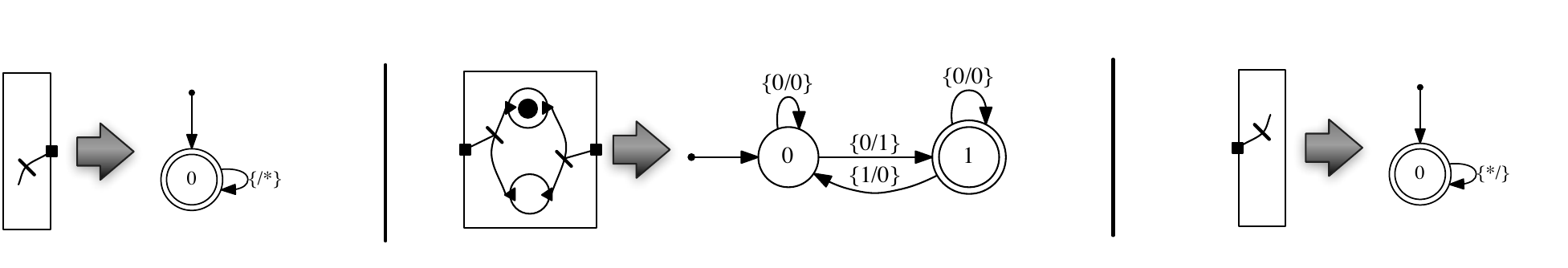}
\]
\caption{Translation to NFAs.\label{fig:piecestonfas}}
\end{figure}

The principle of compositionality, proved in
Theorem~\ref{thm:compositionality}, is illustrated in
\figref{fig:compositionality}: given two $b_1$ nets, we can obtain the NFA
representing their (composite) behaviour in two ways: 1) compose two $b_1$
nets to form the net $b_1\comp b_1$, and then generate its NFA, or
equivalently, 2) generate the two (identical) NFAs for each $b_1$ and compose
them, using a variant\footnote{$(a,b)\dtrans{\alpha}{\beta}(a’,b’)$ iff
$\exists \gamma.\; a\dtrans{\alpha}{\gamma}a’ \mathrel{\wedge}
b\dtrans{\gamma}{\beta}b’$.} of the  product construction.
Compositionality ensures that the diagram commutes, in other words, the
global behaviour of the composition of the two nets is completely determined by
the behaviour of the individual nets, when synchronised along their common
boundary.

%
\begin{figure}
\vspace{-1em}
\[
\includegraphics[height=5cm]{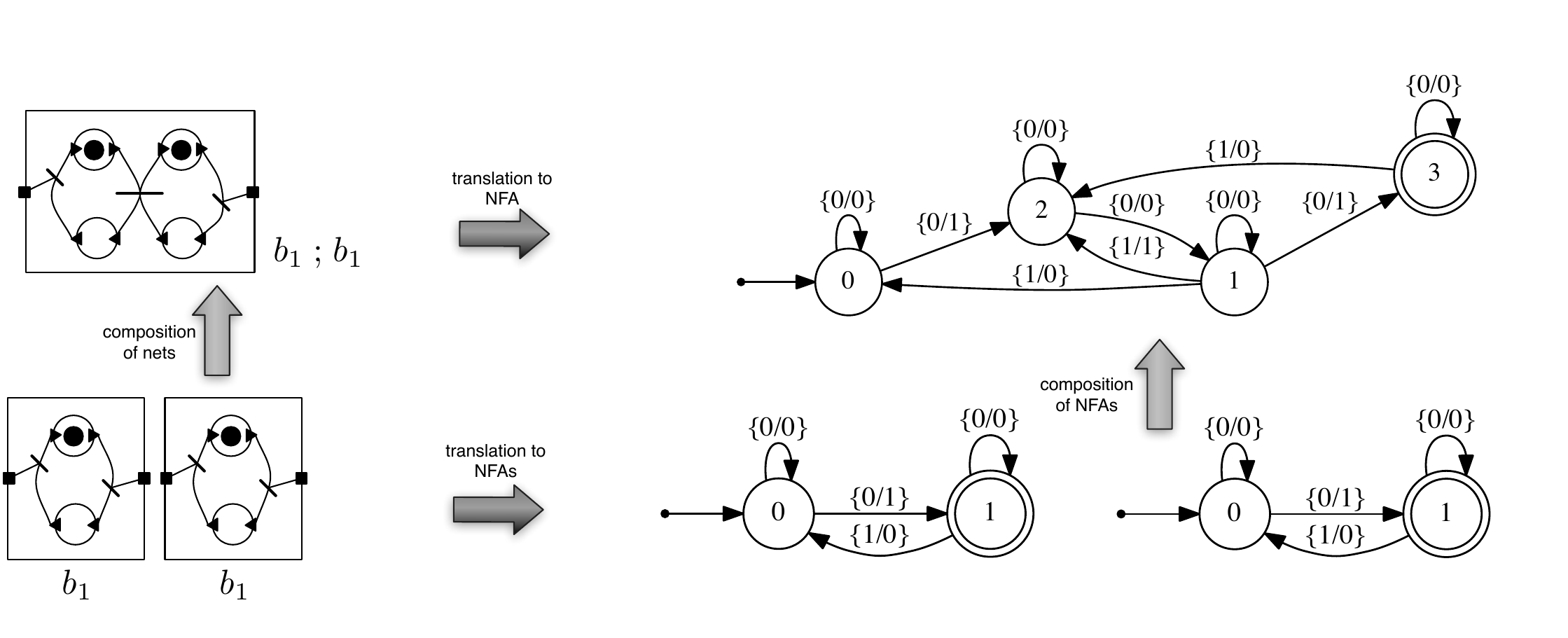}
\]
\caption{\vspace{-1em}Compositionality at work.\label{fig:compositionality}}
\end{figure}
The NFA generated for $b_n=b_1\comp \dots\comp
b_1$ ($b_1$ composed $n$ times) has $2^n$ states, thus directly computing the
automaton for $b_n$ is feasible only for small $n$.
Fortunately, to generate a correct NFA of the composite net, it is sufficient
to capture how each component net must interact on its boundaries in order to
reach its local desired marking---its “synchronisation policy”.
To do this, we close the NFA with
respect to internal ($\epsilon$-) moves---those transitions labelled solely
with $0$s, signifying no interaction at the boundaries---to obtain an automaton
with the same states, but with transitions being paths
$a(\dtrans{0}{0})^*\dtrans{a}{b}(\dtrans{0}{0})^* b$. We then minimise the new
NFA, obtaining a deterministic automaton (DFA), with an “error” state
that is reached whenever an illegal (i.e. not in the behaviour of the
underlying net) interaction is observed on the boundaries. This DFA
minimally represents the entire behaviour (assuming that an observer may only
observe traces) of the net, w.r.t. interactions on its boundaries.

Note that the states of the NFA obtained from a net are 1-1 with the reachable markings of the underlying net; in general, this is not the case after $\epsilon$-closure and minimisation: the states of the minimal DFA merely capture the ``protocol'' the net must follow when interacting with its environment, in order to arrive at the desired marking.
Indeed, for $b_n$, the resulting minimal DFA has $n+2$ states. Of course, computing the minimisation of an NFA can be very expensive---in the worse case, triple exponential in the number of places of the original net---our strategy is thus roughly to decompose nets as far as possible (thereby only minimising small
NFAs) and take advantage of any regular, repetitive structure in the net, via
memoisation. As discussed, compositionality guarantees correctness---the fact
that the square in \figref{fig:minimisation}, illustrating the process for
$B_4$, commutes is a consequence of Theorems~\ref{thm:weakCompositionality}
and~\ref{thm:correctness}.
\begin{figure}
\[
\includegraphics[height=4.5cm]{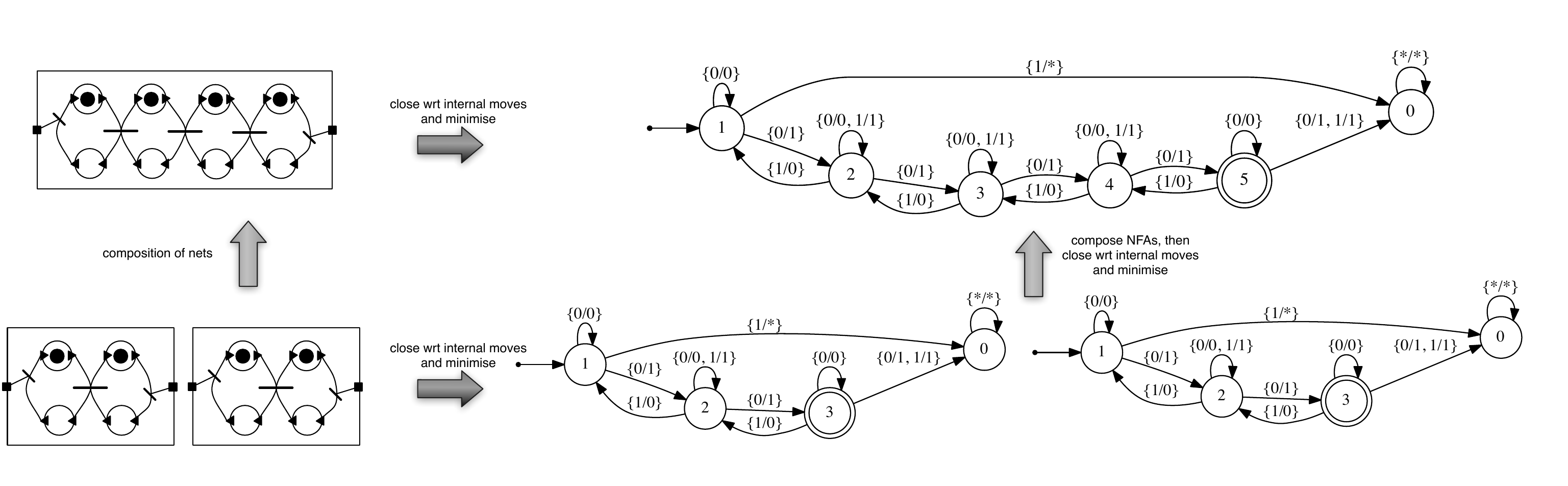}
\]
\caption{Minimising $B_4$, compositionally.\label{fig:minimisation}}
\end{figure}

The applicability of our approach depends on finding ``good'' decompositions of
nets. For $B_n\Defeq \top\comp b_n\comp \bot$, there are many potential decompositions: the
optimal\footnote{All experiments were run on an Intel i7-2600 3.40GHz CPU,
16GB of RAM, running 64-bit Ubuntu Linux.} is the 1st decomposition in
\figref{fig:decompositions}, which corresponds to the algebra term $(\top \comp
(b_1\comp (\dots \comp (b_1\comp \bot)\dots )$. Indeed, the composition of
$b_1$ and $\bot$ minimises to the trivial accepting automaton;
\figref{fig:translation} contains illustrative translation steps of the
different decompositions of $B_4$. In \textit{(i)} the composition of the
automaton for $b_1$ is composed with the automaton for $\bot$: after
minimisation we again obtain the automaton for $\bot$. Thus the procedure
reaches a fixed point after the first step, as illustrated in \textit{(ii)}.
This fact formally captures the intuition about $N_1$ given in
Remark~\ref{rmk:policies}. For this decomposition, memoisation guarantees that
the composition and minimisation is performed only once. In particular, this
means that checking reachability for $B_n$, given this decomposition, is
\emph{linear} in $n$.  However, other decompositions do not lead to such good
performance. In particular, consider the 2nd decomposition of
\figref{fig:decompositions}; here, memoisation does not help (we obtain a
different NFA composition after each step) and we must perform minimisation
after each composition, as illustrated in steps \textit{(iii)} and
\textit{(iv)} of \figref{fig:decompositions}.
\begin{figure}
\[
\includegraphics[height=3cm]{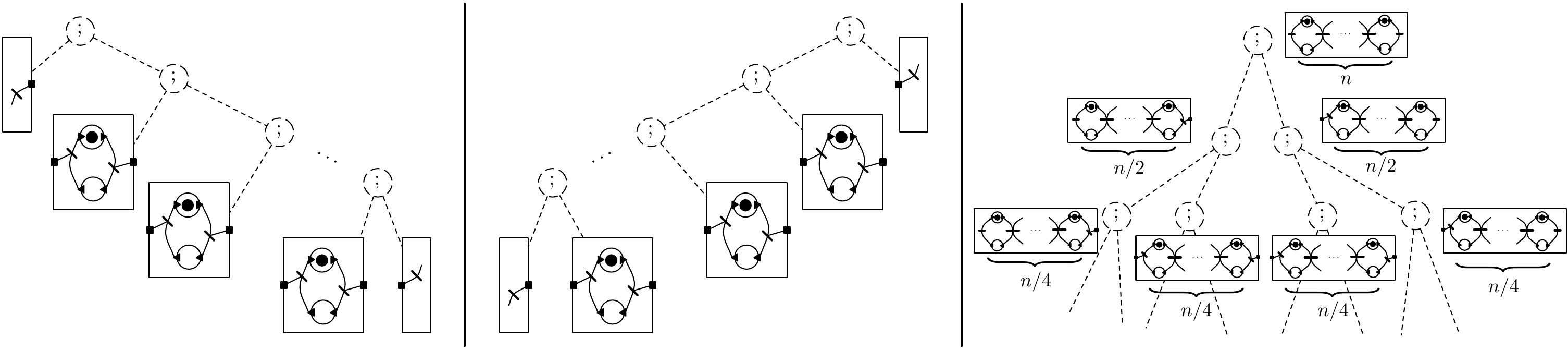}
\]
\caption{Three decompositions of $B_n$, to which we refer as, respectively,
right, left and balanced.\label{fig:decompositions}}
\end{figure}

Our automated approach to deconstructing $B_n$ (discussed in
\S\ref{sec:decomposition}) produces the 3rd (balanced) decomposition of
\figref{fig:decompositions}. In this particular case we decompose by
identifying a transition that connects two components of similar size. This
decomposition, while not optimal, allows frequent use of memoisation, reducing
the amount of computation. A table of running times for the construction of a
minimal DFA for $B_n$, following the three decompositions of
\figref{fig:decompositions}, is given in
\figref{fig:decompositionTimes}.
\begin{figure}
\[
\includegraphics[height=9cm]{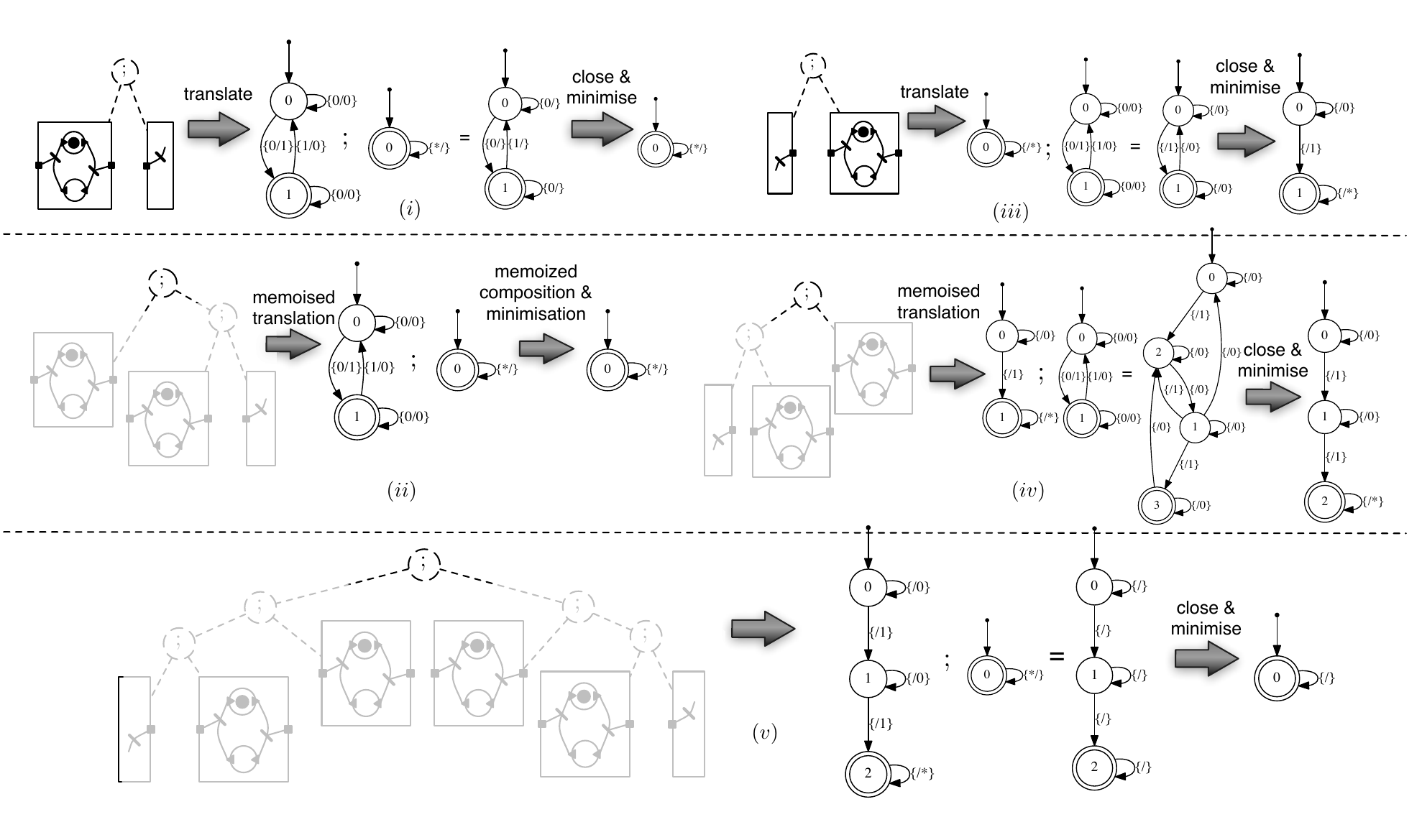}
\]
\caption{Translation of the decompositions in
\figref{fig:decompositions}. (i),(ii) initial steps using the right
decomposition; (iii), (iv) initial steps using the left decomposition; (v)
final step using the balanced decomposition of $B_4$. \label{fig:translation}}
\end{figure}
\begin{figure}
\centering
\begin{subfigure}[t]{0.45\linewidth}
\begin{tabular}{| c || r | r | r | r | }
\hline
& min \# & \multicolumn{3}{|c|}{Time [s]} \\
\cline{3-5}
n     & firing sequence & right & left   & balanced \\ \hline
16    & 136             & 0.000 & 0.020  & 0.008    \\
32    & 528             & 0.000 & 0.140  & 0.024    \\
64    & 2080            & 0.000 & 1.108  & 0.172    \\
128   & 8256            & 0.000 & 12.597 & 2.954    \\
256   & 32896           & 0.000 & -      & 74.737   \\
65536 & 2147516416      & 0.228 & -      & -        \\
\hline
\end{tabular}

\caption{Time to construct minimal DFA for $B_n$ with the three decompositions
illustrated in \figref{fig:decompositions}.\label{fig:decompositionTimes}}
\end{subfigure}%
\hspace{0.5cm}
\begin{subfigure}[t]{0.45\linewidth}
\begin{tabular}{| c || r | r | r | }
\hline
n  & \# places in net & min \# firing sequence & time [s]\\ \hline
6  & 63               & 120                    & 0.028  \\
8  & 255              & 502                    & 0.088  \\
10 & 1023             & 2036                   & 0.312  \\
12 & 4095             & 8178                   & 1.332  \\
14 & 16383            & 32752                  & 5.864  \\
16 & 65535            & 131054                 & 25.014 \\
\hline
\end{tabular}

\caption{Time to construct minimal DFA for $T_n$, using the decomposition described in \figref{fig:treenet}. \label{fig:treeDecompositionTimes}\\}
\end{subfigure}
\caption{NFA construction times for $B_n$ and $T_n$.\label{fig:bnandtnexperiments}}
\end{figure}
%

We have illustrated how the operation `$\comp$' allows
decomposition of the net $B_n$ in order to exhibit its the regular structure.
We will briefly consider a second example that illustrates the use of the
second operation of the algebra, ‘$\otimes$’.
Consider the net in \figref{fig:treenet},
where we want to check whether all the places can be
marked; N.B this net is not 1-safe, but 1-boundedness means that a
transition is blocked if there is a token present in its post-set.
Our automated procedure constructs the decomposition illustrated in the right
part of \figref{fig:treenet}.
In \figref{fig:treeTranslation} we illustrate the steps involved in calculating
the minimal DFA for $T_3$, and give a table of experimental results in
\figref{fig:treeDecompositionTimes}.
\begin{figure}
\[
\includegraphics[height=3cm]{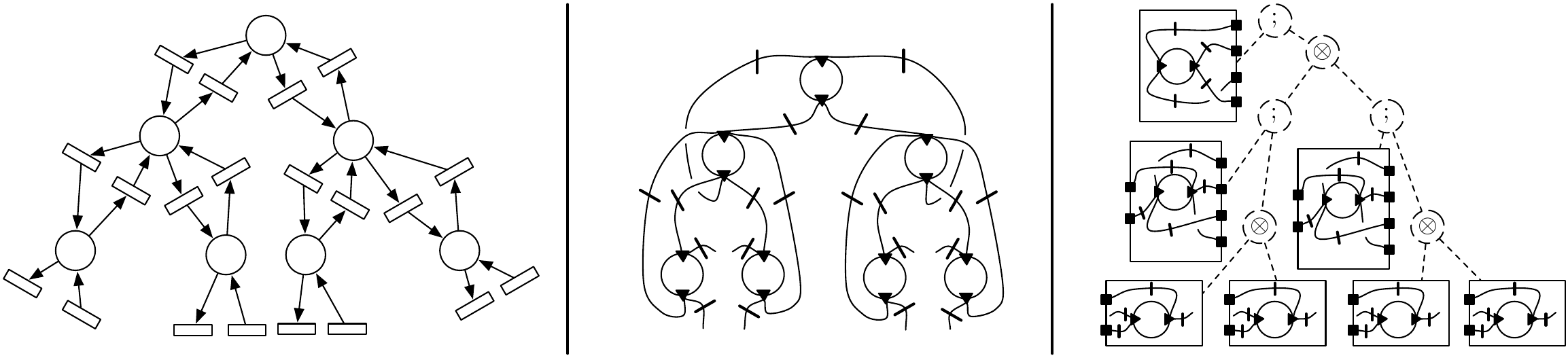}
\]
\caption{The net $T_3$, in traditional and alternative graphical
notation, and its decomposition. \label{fig:treenet}}
\end{figure}

\paragraph{Structure of the paper.}
In \S\ref{sec:theory} we study the foundations of our technique and prove it
correct. In \S\ref{sec:implementation} we discuss our implementation and give
additional experimental results. Connections with related work are in
\S\ref{sec:related}, and we conclude with directions for future research in
\S\ref{sec:discussion}. Due to space constraints, proofs and non-essential
figures have been moved to the appendix.

\section{Nets with boundaries}

\label{sec:theory}
In this section we give the theoretical underpinnings of our technique, harnessing the compositionality of the algebra of nets with boundaries in order to prove its correctness.

\paragraph{Notational conventions.}
For $n\in\N$ let $\underline{n}=\{0,1,\dots,n-1\}$. We write $2^X$ for the powerset of $X$. We write $X+Y$ for the set $\{(x,0)\;|\;x\in X\}\cup\{(y,1)\;|\;y\in Y\}$. Given $\mathcal{U}\subseteq 2^X$ and $\mathcal{V}\subseteq 2^Y$ we write $\mathcal{U}\uplus\mathcal{V}=\{\,U + V\;|\;U\in\mathcal{U},\,V\in\mathcal{V}\,\}\subseteq 2^{X+Y}$. We identify binary strings $\alpha=\alpha_0\alpha_1\dots\alpha_{k-1}$ of length $k$ with subsets of $\underline{k}$ in the obvious way: $\alpha_i=1$ iff $i\in \alpha$.

\begin{definition}
A \emph{net with boundaries} $N:k\to l$ is $(P, T,  k, l, \pre{-}, \post{-},\source{-},\target{-})$ where:
\begin{itemize}
\item[-] $P$ is the set of places, $T$ is the set of transitions
\item[-] $k,l\in\N$ are, respectively, the left and the right boundaries
\item[-] $\pre{-},\post{-}: T \to 2^P$ give, respectively, the pre- and post-sets of each transition
\item[-] $\source{-}:T\to 2^{\underline{k}}$ and $\target{-}:T\to 2^{\underline{l}}$ connect each transition to, resp., the left and the right boundary.
\end{itemize}
Additionally, we assume\footnote{That is, at most one transition can be connected to any place on the boundary. This assumption allows us to simplify the definition of composition of nets; for the more general case see~\cite{Bruni2012}.} that for any $t\neq t'\in T$, $\source{t}\cap\source{t'}=\varnothing$ and $\target{t}\cap\target{t'}=\varnothing$. Ordinary Petri nets can be considered as nets $N:0\to 0$ with no boundaries.\end{definition}

We must use step semantics of nets instead of the more common interleaving semantics to guarantee compositionality; we will illustrate this in Remark~\ref{rmk:step}. Let $\preandpost{t}\Defeq \pre{t}\cup \post{t}$. Transitions $t\neq t'\in T$ are said to be \emph{independent} when $\preandpost{t}\cap\preandpost{t'}=\varnothing$. A set $U\subseteq T$ is said to be \emph{mutually independent} (MI) when for all $u\neq u'\in U$, $u$ and $u'$ are independent. For sets of transitions $U\subseteq T$ we will abuse notation and write $\pre{U}=\bigcup_{u\in U}\pre{u}$, and similarly for $\post{U}$, $\source{U}$ and $\target{U}$.

Each net with boundaries $N:k\to l$ determines an LTS whose transitions witness the step semantics of the underlying net, originally described by Katis et al~\cite{Katis1997b}. For the 1-bounded case, the labels are pairs of binary strings of length $k$ and $l$, respectively.  The states are markings of $N$, denoted by $\marking{N}{X}$, where $X\subseteq P$. The transition relation is defined:
\begin{equation*}
\marking{N}{X} \dtrans{\alpha}{\beta} \marking{N}{X'} \Leftrightarrow\exists\text{ MI }U\subseteq T, \pre{U}\subseteq X,\,\post{U}\cap X=\varnothing,\, X'= (X\backslash \pre{U})\cup \post{U},\, \source{U}=\alpha,\, \target{U}=\beta
\end{equation*}

\subsection{Composition of nets with boundaries}

\label{sec:composition}
Suppose that $N:k\to l$ and $M:l\to m$ are nets with boundaries. A \emph{synchronisation} is a pair $(U,V)$ where $U\subseteq T_N$ and $V\subseteq T_M$  are MI sets of transitions, with $\target{U}=\source{V}$. Given synchronisations $(U,V)$ and $(U',V')$, we say $(U,V)\subseteq (U',V')$ exactly when $U\subseteq U'$ and $V\subseteq V'$. The \emph{trivial synchronisation} is $(\varnothing,\varnothing)$. A synchronisation $(U,V)$ is said to by \emph{minimal} when it is non trivial and, for all synchronisations $(U',V')$, if $(U',V')\subseteq (U,V)$ then $(U',V')$ is trivial. The set of minimal synchronisations of $N$ and $M$ is denoted $\minsynch{N}{M}$. The composed net $N\comp M: k\to m$ has:
\begin{itemize}
\item[-] $P_N + P_M$ as its set of places,
\item[-] $\minsynch{M}{N}$ as its set of transitions. Given $(U,V)\in \minsynch{M}{N}$ we let
$\pre{(U,V)}\Defeq\pre{U}\uplus\pre{V}$, $\post{(U,V)}\Defeq\post{U}\uplus\post{V}$, $\source{(U,V)}\Defeq\source{U}$ and $\target{(U,V)}\Defeq\target{V}$.
\end{itemize}
Examples of compositions of the net $B_n:0\to 0$ are given in Figs.~\ref{fig:bookshelfalternative} and~\ref{fig:decompositions}. Another example is given in \figref{fig:stepComposition}, with the resulting transition arising from the minimal synchronisation $(\{t_1,t_2\},\{t_3\})$.
\begin{figure}
\[
\includegraphics[height=1.8cm]{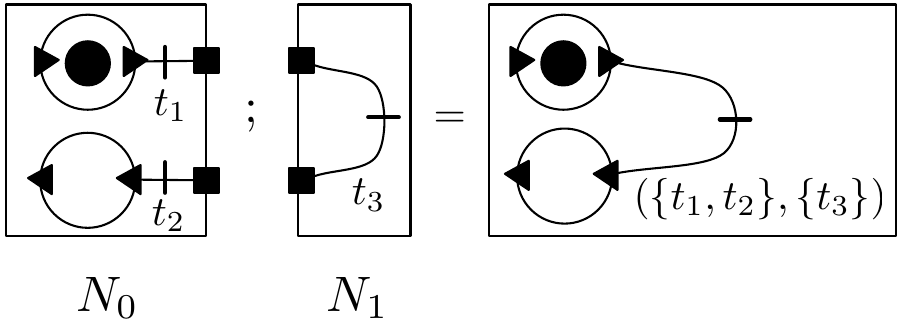}
\]
\caption{Illustration of composition of two nets.\label{fig:stepComposition}}
\end{figure}

\begin{remark}\label{rmk:step}
The example in \figref{fig:stepComposition} illustrates the necessity for step semantics in order for compositionality to hold. Indeed, in the composition $N_0;N_1$ we have the transition $\marking{N_0;N_1}{\{0\}}\dtrans{}{}\marking{N_0;N_1}{\{1\}}$ that witnesses the firing of its transition. This transition decomposes into  $\marking{N_0}{\{0\}}\dtrans{}{11}\marking{N_0}{\{1\}}$ and $\marking{N_1}{\varnothing}\dtrans{11}{}\marking{N_1}{\varnothing}$. The first of these requires the \emph{simultaneous} firing of $t_1$ and $t_2$ in $N_0$; thus if we had considered interleaving semantics then compositionality would fail in this example.
\end{remark}
The next result is a special case of \cite[Theorem~3.6]{Bruni2012}, where a more general algebra of nets is considered. We will rely on this to prove the correctness of our technique in Theorems~\ref{thm:weakCompositionality} and~\ref{thm:correctness}.
\begin{theorem}[Compositionality\label{thm:compositionality}]
Suppose that $N:k\to l$ and $M:l\to m$ are nets with boundaries.  The following holds for all $X,X'\subseteq P_N$, $Y,Y'\subseteq P_M$, $\alpha\in\{0,1\}^k$ and $\beta\in\{0,1\}^m$:
\[
\marking{N\comp M}{X\uplus Y} \dtrans{\alpha}{\beta} \marking{N\comp M}{X'\uplus Y'} \quad\Leftrightarrow\quad
\exists\gamma\in\{0,1\}^l.\;
\marking{N}{X}\dtrans{\alpha}{\gamma}\marking{N}{X'}
\ \wedge\
\marking{M}{Y}\dtrans{\gamma}{\beta}\marking{M}{Y'}
\]
\qed
\end{theorem}
The conclusion of Theorem~\ref{thm:compositionality} implies that, for instance, bisimilarity is a congruence w.r.t. to ‘$\comp$'. For the purposes of reachability checking, traces are sufficient.
\begin{corollary}\label{cor:traces}
There exists a trace $\marking{N\comp M}{X\uplus Y}\dtrans{\alpha_1}{\beta_1}\dots\dtrans{\alpha_p}{\beta_p} \marking{N\comp M}{X'\uplus Y'}$
iff there exist traces $\marking{N}{X}\dtrans{\alpha_1}{\gamma_1}\dots\dtrans{\alpha_p}{\gamma_p} \marking{N}{X'}$ and $\marking{M}{Y}\dtrans{\gamma_1}{\beta_1}\dots\dtrans{\gamma_p}{\beta_p} \marking{N}{X'}$.
\qed
\end{corollary}
In particular, to check for reachability in a composed net, it suffices to find computations in the components that agree on their shared boundary.

\smallskip
The other operation on nets with boundaries is $\otimes$, which can be understood as a parallel composition of nets. Given $N:k\to l$ and $M:m\to n$, $M\otimes N: k+m \to l+n$ has:
\begin{itemize}
\item $P_N+P_M$ as its set of places,
\item $T_N+T_M$ as its set of transitions.
$\pre{(t,0)}\Defeq \{\,(p,0) \;|\; p\in \pre{t}\,\}$,
$\pre{(t,1)}\Defeq \{\,(p,1) \;|\; p\in\pre{t}\,\}$, and similarly for
$\post{(t,0)}$ and $\post{(t,1)}$. Instead
$\source{(t,0)}=\source{t}$ while
$\source{(t,1)}=\{\,k+i\;|\;i\in \source{t}\,\}$;
similarly $\target{(t,0)}=\target{t}$ and
$\target{(t,1)}=\{\,l+i\;|\;i\in \target{t}\,\}$.
\end{itemize}
Compositionality also holds w.r.t. $\otimes$: $\marking{M\otimes N}{X+Y}\dtrans{\alpha\gamma}{\beta\delta} \marking{M\otimes N}{X'+Y'}$ iff $\marking{M}{X}\dtrans{\alpha}{\beta}\marking{M}{X'}$ and $\marking{N}{Y}\dtrans{\gamma}{\delta}\marking{N}{Y'}$. Due to space constraints we omit the details here; they are straightforward as there is no interaction between the two nets.

\subsection{From nets with boundaries to NFAs}
\label{sec:netstonfas}
By an NFA with boundaries $A:k\to l$ we mean an NFA $A$ with set of labels  $\{0,1\}^k\times\{0,1\}^l$, written $\alpha/\beta$, where $\alpha\in\{0,1\}^k$ and $\beta\in\{0,1\}^l$. Given NFA with boundaries $A:k\to l$ and $B:l\to m$, the NFA with boundaries $A\comp B:k\to m$ is obtained by a variant of the product construction where $(x,y)\dtrans{\alpha}{\beta}(x',y')$ iff there exists $\gamma\in\{0,1\}^l$ such that $x\dtrans{\alpha}{\gamma} x'$ and $y\dtrans{\gamma}{\beta}y'$. Given NFA with boundaries $A:k\to l$ and $B:m\to n$, the NFA with boundaries $A\otimes B:k+m\to l+n$ is obtained via another variant of the product construction: here $(x,y)\dtrans{\alpha\gamma}{\beta\delta}(x',y')$ iff $x \dtrans{\alpha}{\beta} x'$ and $y\dtrans{\gamma}{\delta} y'$. The algebra of automata with boundaries described above is an instance of Span(Graph)~\cite{Katis1997a}.

Given a net with boundaries $N:k\to l$, and non-empty sets $\mathcal{X},\mathcal{Y}\subseteq 2^{2^{P_N}}$ of, respectively, \emph{initial} and \emph{final} markings, we can consider its labelled transition system as an NFA, written $\NFA{N}{\mathcal{X}}{\mathcal{Y}}$, that has initial states $\mathcal{X}$ and final states $\mathcal{Y}$. If $N:k\to l$ does not have any places then $\NFA{N}{\{\varnothing\}}{\{\varnothing\}}$ has exactly one state, which is an accept state (see NFA for $\top$, $\bot$ in \figref{fig:piecestonfas}). The following is immediate.
\begin{proposition}
Given $N:k\to l$, initial and final markings $\mathcal{X},\,\mathcal{Y}$, a marking in $\mathcal{Y}$ is reachable from a marking in $\mathcal{X}$ iff $L(\NFA{N}{\mathcal{X}}{\mathcal{Y}})\neq\varnothing$. \qed
\end{proposition}
We also have the following as an immediate consequence of Theorem~\ref{thm:compositionality}:
\[
\NFA{N\comp M:k\to m}{\mathcal{X}\uplus \mathcal{X'}}{\mathcal{Y}\uplus\mathcal{Y'}} \cong (\NFA{N:k\to l}{\mathcal{X}}{\mathcal{Y}}) \comp (\NFA{M:l\to m}{\mathcal{X'}}{\mathcal{Y'}})
\]
and in particular the two automata accept the same language.

\subsection{Weak closure and minimisation}
\label{sec:weakclosure}

Hiding internal computations in individual component nets is crucial for the performance of our technique. The procedure is akin to the  $\tau$-reflexive-transitive closure of an LTS $L$, which yields an LTS $L'$ on which bisimilarity agrees with weak-bisimilarity on $L$, in the sense of Milner~\cite{Milner1989}.

Let $\epsilon_{k,l}=0^k/0^l$. Sometimes we will write simply $\epsilon$ when $k$ and $l$ are clear from the context. Notice that given any net $N:k\to l$, for each marking $X$ there is a transition $\marking{N}{X}\epstrans{k}{l}\marking{N}{X}$ that arises from firing the empty set of net-transitions. In general, transitions $\marking{N}{X}\epstrans{k}{l}\marking{N}{X'}$ witness the firing of ``internal'' net-transitions in $N$, ie those that are not connected to any boundary port.

The \emph{weak} transition system induced by $N:k\to l$
has transitions:
\begin{equation}\label{eq:weakSemantics}
\marking{N}{X} \;\dTrans{\alpha}{\beta}\; \marking{N}{X'} \ \Leftrightarrow\
\exists X'', X'''.\; \marking{N}{X} (\epstrans{k}{l})^* \marking{N}{X''},\,
                          \marking{N}{X''} \dtrans{\alpha}{\beta} \marking{N}{X'''},\,
                          \marking{N}{X'''} (\epstrans{k}{l})^* \marking{N}{X'}
\end{equation}
Note that the above notion of weak transition differs from that considered in~\cite{Bruni2012} but is close to the weak transitions of~\cite{Soboci'nski2009a}.

\begin{theorem}[Compositionality w.r.t. weak semantics\label{thm:weakCompositionality}]
Suppose that $N:k\to l$ and $M:l\to m$ are nets with boundaries.  Then for all $X,X'\subseteq P_N$, $Y,Y'\subseteq P_M$, $\alpha\in\{0,1\}^k$, $\beta\in\{0,1\}^m$:
\begin{enumerate}[(i)]
\item if $\marking{N\comp M}{X + Y} \dTrans{\alpha}{\beta} \marking{N\comp M}{X' + Y'}$
then $\exists\,p,q\in\N$,
$\gamma,\gamma_i,\gamma_j'\in \{0,1\}^l$
for $1\leq i \leq p$ and $1\leq j\leq q$
\[
\!\!\!\!\!\!\!\!\!\!\!\!\!
\marking{N}{X}\dtrans{\!0^k\!}{\!\gamma_1\!}\dots\dtrans{\!0^k\!}{\!\gamma_p\!}
 \dtrans{\!\alpha\!}{\!\gamma\!}
\dtrans{\!0^k\!}{\!\gamma'_1\!}\dots \dtrans{\!0^k\!}{\!\gamma'_q\!}
\marking{N}{X'}
 \text{ and }
\marking{M}{Y}\dtrans{\!\gamma_1\!}{\!0^m\!}\dots\dtrans{\!\gamma_p\!}{\!0^m\!}
 \dtrans{\!\gamma\!}{\!\beta\!}
\dtrans{\!\gamma'_1\!}{\!0^m\!} \dots \dtrans{\!\gamma'_q\!}{\!0^m\!}
\marking{M}{Y'}.
\]
\item if $\marking{N}{X}\dTrans{\alpha}{\gamma} \marking{N}{X'}$ and $\marking{M}{Y}\dTrans{\gamma}{\beta}\marking{M}{Y'}$ for some
$\gamma\in\{0,1\}^l$
then $\marking{N\comp M}{X + Y}\dTrans{\alpha}{\beta} \marking{N\comp M}{X'+ Y'}$. \qed
\end{enumerate}
\end{theorem}

Given an NFA with boundaries $A:k\to l$, let $\epsmin{A}:k\to l$ denote the DFA obtained by $\epsilon_{k,l}$-closure and minimisation.
\begin{remark}
Recall that any ordinary net $N$ can be considered as a net with boundaries $N:0\to 0$. Now $\epsmin{\NFA{N}{\mathcal{X}}{\mathcal{Y}}}:0\to 0$ is one of two DFAs: the DFA with one accept state (if a marking in $\mathcal{Y}$ is reachable from some marking in $\mathcal{X}$) and the DFA with one non-accept state (if no markings in $\mathcal{Y}$ are reachable from any marking in $\mathcal{X}$).
\end{remark}
Given an ordinary Petri net $N$, initial markings $\mathcal{X}$ and final markings $\mathcal{Y}$, a simple but extremely inefficient way of checking the reachability of a marking is thus to directly compute $\epsmin{\NFA{N}{\mathcal{X}}{\mathcal{Y}}}$ and check whether the single state in the resulting DFA is an accept state. Our technique for checking reachability is based on computing this DFA using a structural decomposition of $N$, which, when combined with memoisation, can result in fast execution times.

\subsection{Correctness}
\label{sec:correctness}

Here we give a formal account of our technique and prove it correct, using the previous results in this section. A  \emph{wiring expression} is a syntactic term formed from the following grammar
\[
T \bnfEq x \bnfSep T \comp T \bnfSep T \otimes T
\]
where the leaves $x$ are variables. A \emph{variable assignment} $\mathcal{V}$ is a map that takes variables to nets with boundaries. Given a pair $(t,\,\mathcal{V})$ of a wiring expression $t$ and variable assignment $\mathcal{V}$, its semantics $\semanticsOf{t}_{\mathcal{V}}$ is a net with boundaries, defined recursively in the obvious way: $\semanticsOf{x}_\mathcal{V} \Defeq \mathcal{V}(x)$, $\semanticsOf{t_1 \comp t_2}_\mathcal{V} \Defeq \semanticsOf{t_1}_\mathcal{V} \comp \semanticsOf{t_2}_\mathcal{V}$ and $\semanticsOf{t_1\otimes t_2}_\mathcal{V} \Defeq \semanticsOf{t_1}_\mathcal{V}\otimes\semanticsOf{t_2}_\mathcal{V}$. We implicitly assume that variable assignments are compatible with $t$, in the sense that only nets with a common boundary are composed; we leave out the formal details, which are straightforward.
Given a net $N:k\to l$, we say that $(t,\,\mathcal{V})$ is a \emph{wiring decomposition} of $N$ if $\semanticsOf{t}_\mathcal{V}\cong N$.

Given a wiring decomposition $(t,\,\mathcal{V})$ of $N:k\to l$, together with maps $\mathcal{I}$, $\mathcal{F}$ called, respectively, \emph{initial markings} and \emph{final markings}, that take each variable $x$ to a set of markings of the net $\mathcal{V}(x)$, define $\epsminpr{t}_{(\mathcal{V},\mathcal{I},\mathcal{F})}$ recursively:
\begin{align*}
\epsminpr{x}_{(\mathcal{V},\mathcal{I},\mathcal{F})}
    &\Defeq \epsmin{\NFA{\mathcal{V}(x)}{\mathcal{I}(x)}{\mathcal{F}(x)}},\,\\
 \epsminpr{t\comp t'}_{(\mathcal{V},\mathcal{I},\mathcal{F})} &\Defeq \epsmin{ \epsminpr{t}_{(\mathcal{V},\mathcal{I},\mathcal{F})} \comp \epsminpr{t'}_{(\mathcal{V},\mathcal{I},\mathcal{F})}},\,\\
\epsminpr{t \otimes t'}_{(\mathcal{V},\mathcal{I},\mathcal{F})} &\Defeq \epsmin{ \epsminpr{t}_{(\mathcal{V},\mathcal{I},\mathcal{F})} \otimes \epsminpr{t'}_{(\mathcal{V},\mathcal{I},\mathcal{F})} }.
\end{align*}
 The function $\epsminpr{-}$ is the formalisation of our approach, taking a wiring decomposition, together with initial and final markings to a minimal DFA. 
 Sets of markings of the leaf nets given by $\mathcal{I}$ and $\mathcal{F}$ can be combined to form a set of markings $\mrk{t}{\mathcal{I}}$ of $\semanticsOf{t}_{\mathcal{V}}$ in an obvious way: $\mrk{x}{\mathcal{I}} \Defeq \mathcal{I}(x)$, $\mrk{t\comp t'}{\mathcal{I}} \Defeq \mrk{t}{\mathcal{I}} \uplus \mrk{t'}{\mathcal{I}}$,  $\mrk{t\otimes t}{\mathcal{I}} \Defeq \mrk{t}{\mathcal{I}} \uplus \mrk{t'}{\mathcal{I}}$ (and similarly for $\mathcal{F}$.)


\begin{theorem}[Correctness]\label{thm:correctness}
Suppose $(t,\mathcal{V})$ is a wiring decomposition of $N:k\to l$, $\mathcal{I}$  initial markings and $\mathcal{F}$ final markings. Then
$\epsminpr{t}_{(\mathcal{V},\mathcal{I},\mathcal{F})}
    \cong
        \epsmin{ \NFA{\semanticsOf{t}_\mathcal{V}}{\mrk{t}{\mathcal{I}}}{\mrk{t}{\mathcal{F}} } }$.\qed
\end{theorem}
An example application of Theorem~\ref{thm:correctness} is the commutativity of the diagram in \figref{fig:minimisation}.

Note that we have not discussed how to obtain a wiring decomposition, starting from a net $N:k\to l$. As demonstrated in \figref{fig:decompositionTimes}, different decompositions result in markedly different performance. Our automated procedure for obtaining a decomposition is described in \S\ref{sec:decomposition}.

%

\section{Implementation and experimental results}
\label{sec:implementation}

Our implementation has been written in Haskell, and is available for
download\footnote{\url{http://users.ecs.soton.ac.uk/os1v07/ICALP13}}.
The high level view of our algorithm is:
\begin{enumerate}
    \item As input, take an ordinary marked net $N$ (considered as a net with
          boundaries $N:0\to 0$) and a target marking, given place-wise, to be
          checked for reachability. Concretely, each place is labelled with
          `Yes,' (token must be present) `No' (token must be absent) or `Don't
          care.'
    \item Using an automatic decomposition procedure (described in
          \S\ref{sec:decomposition}), we decompose the net, obtaining a
          wiring decomposition (as introduced in
          \S\ref{sec:correctness}) enhanced with additional information to enable memoisation.
    \item Taking advantage of memoisation---to eliminate duplicate
          computations---traverse the wiring decomposition tree to compute $\epsminpr{-}$:
    \begin{enumerate}
        \item At leaves, we have (typically, small) nets with boundaries, and
              the local desired marking. We use the procedure described in
              \S\ref{sec:netstonfas}
              to generate the NFA that corresponds to the
              net and apply $\epsilon-$closure and minimisation, described in
              \S\ref{sec:minimisation}.
        \item At a composition node, we generate the NFAs corresponding to each
              sub-tree, and compose them using the variant of
              product-construction discussed in \S\ref{sec:intro}, finally
              $\epsilon-$closing and minimising the resulting NFA.
        \item At a tensor node, we generate the NFAs corresponding to each
              sub-tree, combine them using the standard product
              construction on NFAs, and perform minimisation.
    \end{enumerate}
\end{enumerate}
The experimental results given in Figs.\ \ref{fig:bnandtnexperiments} and
\ref{fig:philotimes} are given for pre-constructed decompositions, that is,
only step 3 of the algorithm is performed. The results in
\figref{fig:tndecomposetimes} were obtained using the implementation of the full algorithm.
\subsection{Decomposer}
\label{sec:decomposition}

Our net decomposition algorithm attempts to find decompositions via two
simple approaches: first we look for a net-transition that, when removed,
results in two disconnected nets. If many such transitions exist then we take
the one that results in the most balanced (in number of places) decomposition.
An example is the balanced decomposition in \figref{fig:decompositions}. If
such a transition cannot be found, we look for a place that, once removed,
results in two disconnected nets. This results in a `$\comp$' node (that
results from removing the place) followed by a `$\otimes$' node (that composes
the two disconnected nets). Again, if many such places exist, we look for one
that results in the most balanced decomposition. An example of this
decomposition strategy is the decomposition in \figref{fig:treenet}. Both
searches are quadratic in the size of the net. If neither a suitable transition
nor place is found, we remove a place that results in the smallest boundary,
after decomposition. The time taken to decompose the net $T_n$ is given in
\figref{fig:tndecomposetimes}; in this example
the time to decompose the net dominates. Note that, given a net, a
decomposition must be computed (or given as input) only once, whence 
different various initial markings and desired markings can be
considered.



\subsection{NFA $\epsilon-$closure and minimisation}
\label{sec:minimisation}

Our approach relies on ignoring internal computations to reduce the state space
to be explored. To produce minimal DFAs for an input NFA, we apply epsilon
closure, and minimisation, as detailed in \S\ref{sec:weakclosure}. We perform
epsilon closure through a variant of the subset-construction on NFAs, which
constructs the NFA of sets of states reachable through $\epsilon-$ or standard
transitions, starting from the $\epsilon-$closure of the initial states of the
input NFA. To perform minimisation we employ the well-known algorithm of
Brzozowski~\cite{Brzozowski1962}.

A notable implementation detail is that we use a variant of Reduced Ordered
Binary Decision Diagrams (ROBDD, commonly written as BDD) to encode the
transition relation of the NFA---the labels of our transitions are binary
strings and thus any state $x\in X$ gives rise to a function $\{0,1\}^{k+l}\to \mathcal{P}(X)$. Traditionally, BDDs are used to provide compact representations for functions $\{0,1\}^n\to \{0,1\}$, but we found it a straightforward exercise to generalise from the boolean algebra of the booleans to the boolean algebra of subsets.


\subsection{Experimental results and discussion}

In addition to the results in \figref{fig:bnandtnexperiments} we considered a standard net encoding of the dining philosopher problem. Given the nets in \figref{fig:philocomponents}, let $PhRow_1 \Defeq (ph \comp fk)$, $PhRow_{k+1} \Defeq (ph \comp (fk \comp  PhRow_k))$. Then a table of $n$ dining philosophers can be obtained as:
\begin{equation}\label{eq:philoexpression}
Ph_n\Defeq d_3\comp ((i_3 \otimes PhRow_n) \comp e_3) \qquad (\text{see  \figref{fig:philodecomp}}).
\end{equation}
Running times, when checking for deadlock in $Ph_n$, are given in
\figref{fig:philotimes}. The slow growth w.r.t. $n$ illustrates the fact that
our technique works well when a fixed point is quickly reached when traversing
a wiring decomposition, for example, the right decomposition of $B_n$ in
\figref{fig:decompositions} reaches a fixed point after one `$\comp$' node in
the wiring decomposition.
The fixed point for \eqref{eq:philoexpression} is reached when calculating $PhRow_3$: the resulting minimal DFA has 10 states, as shown in \figref{fig:philofixedpoint}. Intuitively, this means that while one can distinguish between 1, 2 and $\geq$3 philosophers via interaction on the boundary,
 all $PhRow_k$ reduce to the same minimal DFA for $k\geq 3$. Our
procedure takes advantage of this: memoisation of compositions means that we
minimise only once.
%
\begin{figure}
\[
\includegraphics[height=2.5cm]{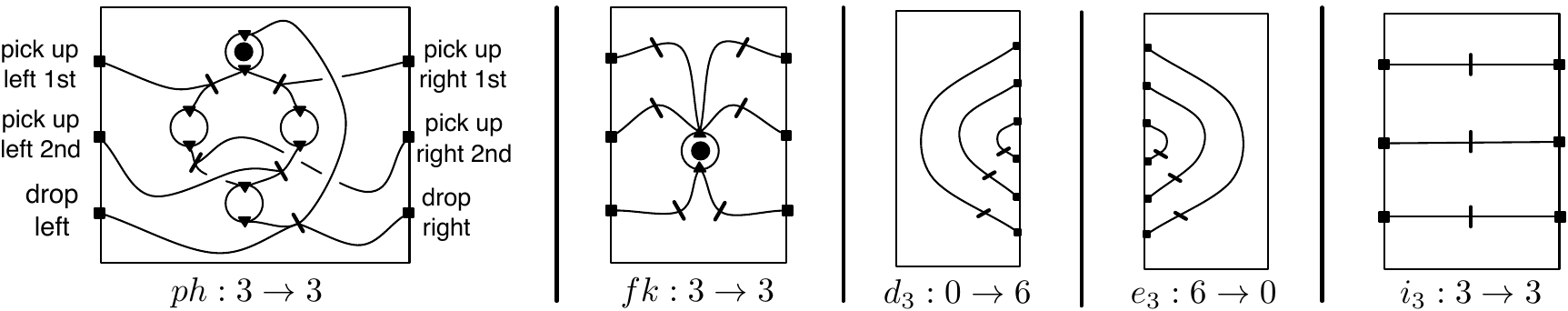}
\]
\caption{Component nets of philosopher decomposition.}
\label{fig:philocomponents}
\end{figure}
\begin{figure}
\centering
\begin{subfigure}[t]{0.45\linewidth}
    \begin{tabular}{| c || r | r |}
\hline
n & deconstruction [s] & NFA generation [s] \\ \hline
4  & 0.052   & 0.008 \\
5  & 0.240   & 0.008 \\
6  & 1.108   & 0.004 \\
7  & 5.104   & 0.008 \\
8  & 23.261  & 0.008 \\
9  & 103.106 & 0.012 \\
10 & 451.628 & 0.012 \\
\hline
\end{tabular}

    \caption{Time to deconstruct $T_n$ (as per
    \figref{fig:treenet}) and generate the minimal DFA.}
\label{fig:tndecomposetimes}
\end{subfigure}%
\hspace{0.5cm}
\begin{subfigure}[t]{0.45\linewidth}
\centering
    \begin{tabular}{| c || r | r | r | }
\hline
n    & time [s] \\ \hline
1    & 2.072 \\
4    & 3.844 \\
16   & 3.924 \\
64   & 3.920 \\
128  & 3.908 \\
256  & 3.908 \\
1024 & 3.684 \\
\hline
\end{tabular}

\caption{Time to generate minimal DFA for $Ph_n$, defined in~\eqref{eq:philoexpression}.}
\label{fig:philotimes}
\end{subfigure}
\caption{Example NFA construction times for $B_n$ and $T_n$.}
\end{figure}

Many nets are not amenable to efficient decomposition and are unsuitable for
our technique. For instance, our implementation performs poorly when input nets are cliques, nets where every place is connected to every other by a transition, or in general, on ``densely connected'' nets.
One reason why our technique is infeasible for such nets is because
two factors influence the size of the generated NFA from a net $N:k\to l$: \textit{(i)} the number of places---if $N$ has $n$ places, this can translate to potentially $2^n$ NFA-states, and \textit{(ii)} the size of the net boundaries, since it implies an alphabet of size up to $2^{(k + l)}$. In fact, even with hand constructed decompositions, our implementation fails to terminate even for very small cliques due to large boundaries in any decomposition.


\section{Related work}
\label{sec:related}

\paragraph{Algebras of nets and automata.}
The algebra of automata with boundaries used in this paper is an instance of
the algebra of Span(Graph)~\cite{Katis1997a}, developed by {R.F.C.} Walters and
collaborators: in fact, a translation from nets to this algebra was already
present in~\cite{Katis1997b}. The goal of the more recent
work~\cite{Soboci'nski2010, Bruni2011, Bruni2012} was to lift this algebra
to the level of nets in a compositional way, study the resulting behavioural equivalences and explore connections with process algebra. A theme of our work is to ignore state and focus on external interactions: here we were inspired by the ideas of Milner~\cite{Milner1989}. Conceptually related
approaches in semantics of programming languages
include~\cite{Reddy1996,Ghica2003}.

\paragraph{Reachability} in bounded, finite state Petri nets is a
widely-studied problem and there are several well-known approaches to
mitigating the impact of state-explosion (it follows from~\cite{Cheng1993} that
the problem is PSPACE-complete.) Due to space constraints we are able to offer
only cursory overviews and comparisons of techniques that are most related to
our approach. A well-known technique is partial order reduction: in a seminal
paper, McMillan~\cite{McMillan1995} used the unfolding
construction~\cite{Nielsen1980} in order to analyse reachability in Petri nets
by generating finite complete prefixes, that is, initial parts of unfoldings
that suffice for reachability. The algorithm to compute the finite complete
prefix was later improved~\cite{Esparza2002,Khomenko2003}. Unfoldings (and
finite complete prefixes) carry more information about the computations of nets
than merely reachability, for instance, allowing LTL model
checking~\cite{Esparza2001}. For an overview of the extensive field
see~\cite{Esparza2008}. A finite complete prefix must be constructed prior to a
reachability analysis, analogously to our construction of a wiring
decomposition prior to translation. Because of the different nature of the two
approaches, it is difficult to offer a thorough analysis of the relative
performance of the two approaches: on some of the examples we have considered
the performance of our implementation is competitive (compare
\figref{fig:decompositionTimes} with~\cite[Table 1]{Esparza2002}.) Another
technique, known as symmetry reduction~\cite{Starke1991,Schmidt2000}, exploits
symmetries in the state space: the goal is, roughly, to build a reduced
reachability graph in order to visit only one representative from each orbit.
Our use of memoisation is similar in spirit to symmetry reduction, since we
only need to translate any particular wiring decomposition once.

In experiments ($B_n$, $T_n$, $Ph_n$ and others) our implementation often performs well in identifying unreachable configurations;
this is because in many systems the reasons for a configuration being
unreachable are ``local''. Here our approach contrasts with techniques such as
unfolding or symmetry reduction where (efficient representations of) explicit
reachability graphs are constructed.

\section{Conclusions and future work}

\label{sec:discussion}
We have introduced a new technique for reachability in bounded Petri nets, based on \textit{(i)} structural decomposition using a recently developed compositional algebra and \textit{(ii)} avoiding state explosion by focusing only on interactions between component nets, forgetting internal state.
Our technique depends on finding efficient decompositions and works best when
the computation reaches fixpoints w.r.t. interactions on boundaries in composed systems, as illustrated in the examples that we have highlighted. We have proved that the technique is correct, implemented it and performed a number of experiments. Finally, we have developed and implemented an algorithm for automatic decomposition of nets that performs adequately on a number of examples.

In future work we plan to improve our decomposition algorithm and characterise the class of nets to which our approach is suited. Additionally, by using the full algebra~\cite{Soboci'nski2010,Bruni2012} of nets, in particular, the possibility of connecting several transitions to the same boundary port, we hope to alleviate some of the problems identified in \S\ref{sec:implementation}.
We also plan to generalise our approach to other models: for example by
examining symbolic representations of the algebras of P/T nets in~\cite{Bruni2011,Bruni2012} we hope to extend our technique to coverability.






{
\bibliography{../jab}
\bibliographystyle{abbrv}
}
\newpage
\section*{Appendix}

\begin{figure}
\[
\!\!\!\!\!\!\!\!\!\includegraphics[height=17cm]{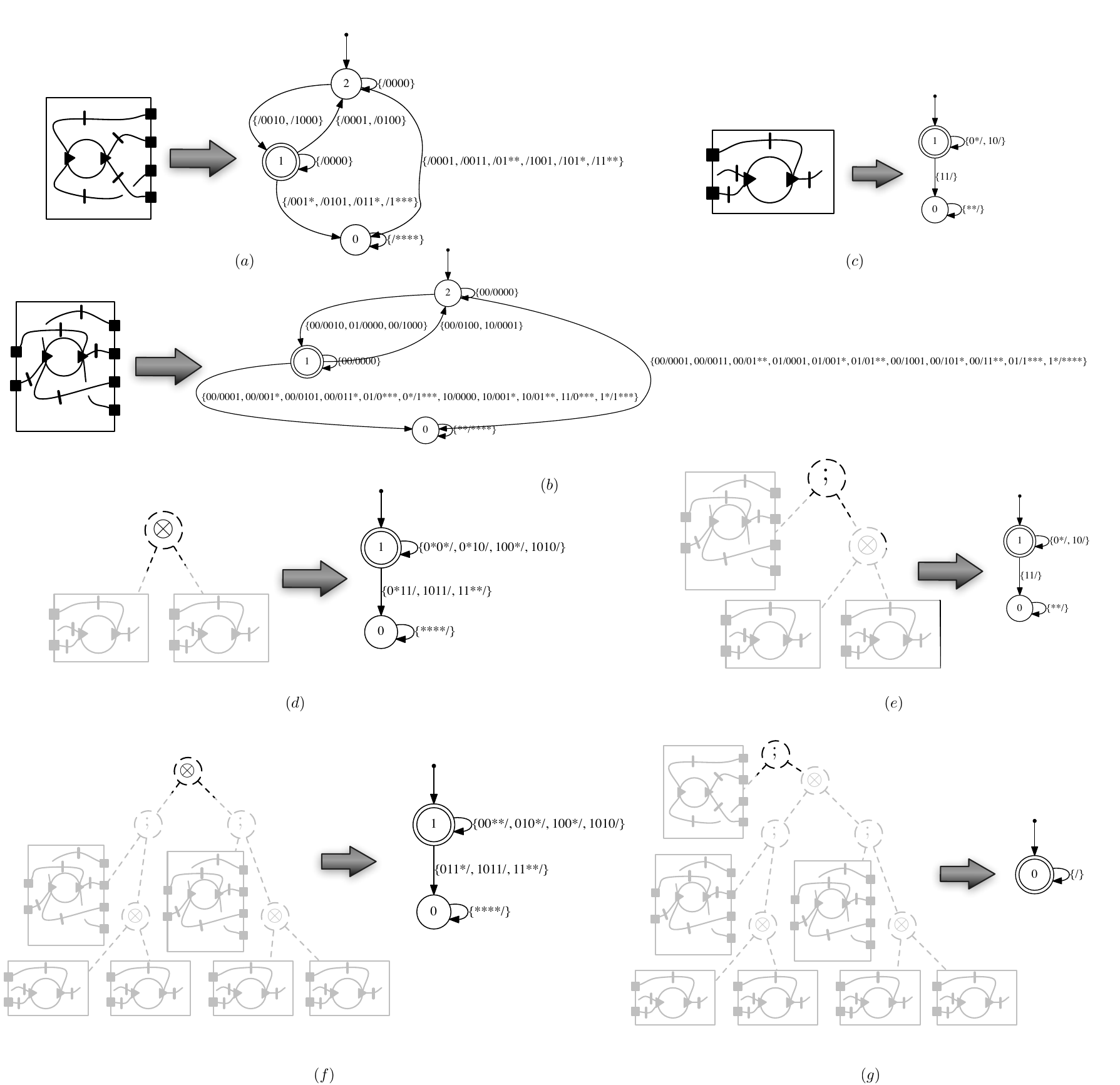}
\]
\caption{Steps involved in translating $T_3$ to an NFA.\label{fig:treeTranslation}}
\end{figure}

\begin{figure}
\label{fig:diningphilos}
\[
\includegraphics[height=8cm]{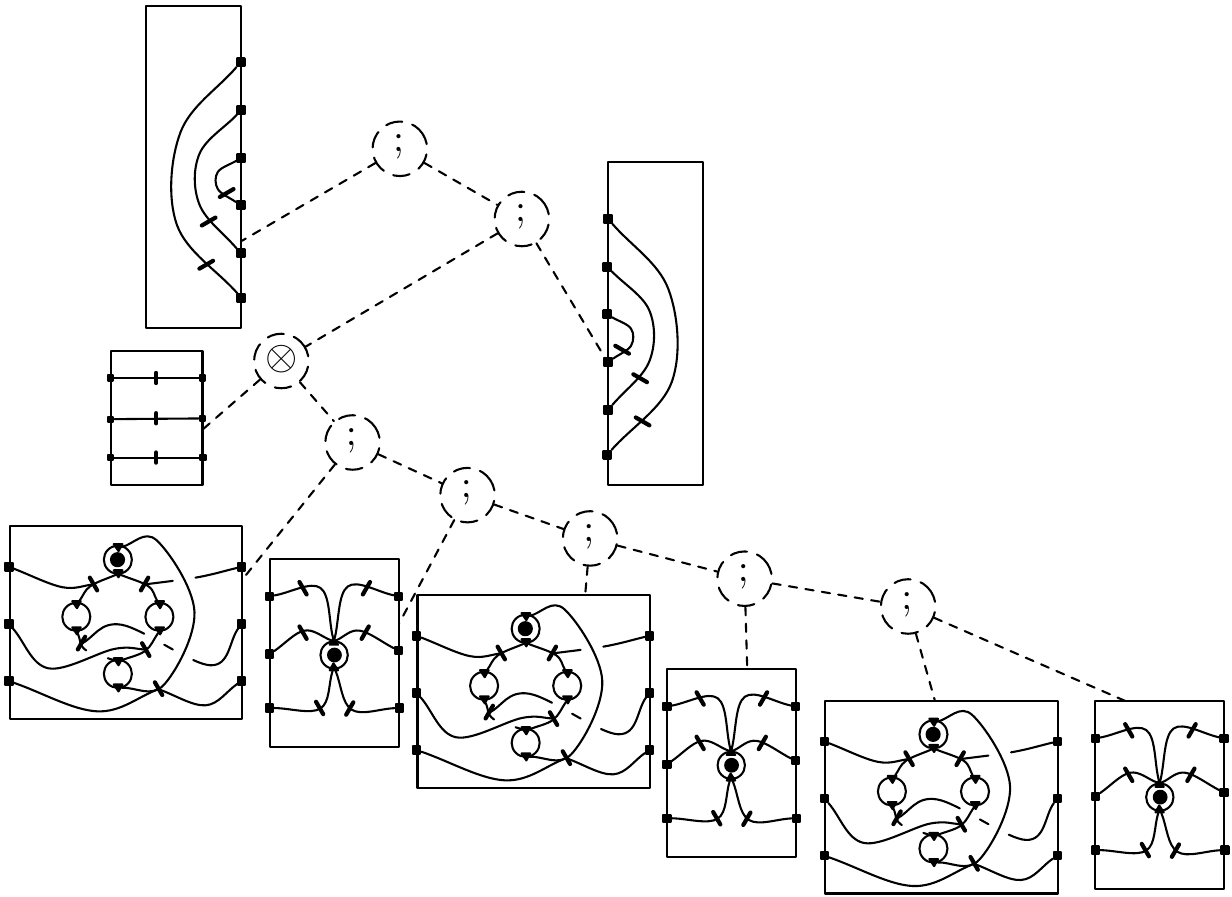}
\]
\caption{Decomposition of three dining philosophers.\label{fig:philodecomp}}
\end{figure}

In order to prove compositionality we first need to prove a small, technical lemma.
\begin{lemma}\label{lem:minimalDecomposition}
Suppose that $N:k\to l$ and $M:l\to m$ are nets with boundaries and $(U,V)$ is a non-trivial synchronisation. Then there exists a mutually independent family $\{(U_i,V_i)\}_{i\in I}$ of minimal synchronisations with $U=\bigcup_{i\in I}U_i$ and $V=\bigcup_{i\in I} V_i$.
\end{lemma}
\begin{proof}
We argue by induction on $|U+V|$. If $(U,V)$ is minimal then the singleton family $\{(U,V)\}$ satisfies the requirements. Otherwise there exists a minimal synchronisation $(U',V')\subseteq (U,V)$. Now since there is at most one transition connected to each point on the boundary, we have $\target{U'}\cap \target{(U\backslash U')}=\varnothing$ and, similarly, $\source{V'}\cap\target{(V\backslash V')}=\varnothing$. Since $\target{U}=\source{V}$, we must also have $\target{(U\backslash U')}=\source{(V\backslash V')}$ and thus $(U\backslash U',V\backslash V')$ is a synchronisation. By the inductive hypothesis, there exists a mutually independent family $\{(U_i, V_i)\}_{i\in I}$, and so $\{(U',V')\}\cup \{(U_i, V_i)\}_{i\in I}$ fulfils the requirements.  \qed
\end{proof}

\paragraph{Proof of Theorem~\ref{thm:compositionality}.}
\begin{proof}
($\Rightarrow$) If $\marking{N\comp M}{X\uplus Y} \dtrans{\alpha}{\beta} \marking{N\comp M}{X'\uplus Y'}$ then there exists mutually independent set of minimal synchronisations $W\subseteq \minsynch{N}{M}$ with $\source{W}=\alpha$ and $\target{\alpha}=\beta$. Consider $U\Defeq \bigcup_{(X,Y) \in W} X \subseteq T_N$ and $V\Defeq \bigcup_{(X,Y)\in W} Y \subseteq T_M$. Since each $(X,Y)\in W$ is a synchronisation, we have $\target{X}=\source{Y}$ and so $\target{U}=\source{V}$. By definition, in each $(X,Y)\in W$, $X$ and $Y$ are mutually independent in, respectively, $N$ and $M$. Since $W$ is mutually independent, if $(X,Y)\neq (X',Y')\in W$ we have $\preandpost{(X,Y)}\cap \preandpost{(X',Y')}=\varnothing$, so $(\preandpost{X}+\preandpost{Y}) \cap  (\preandpost{X'}+\preandpost{Y'}) = \varnothing$ and thus both $\preandpost{X}\cap \preandpost{X'} =\varnothing$ and $\preandpost{Y}\cap \preandpost{Y'}=\varnothing$. It follows that $U$ and $V$ are mutually independent in $N$ and $M$, respectively, and letting $\gamma\Defeq \target{U}(=\source{V})$ we have  $\marking{N}{X}\dtrans{\alpha}{\gamma}\marking{N}{X'}$ and $\marking{M}{Y}\dtrans{\gamma}{\beta}\marking{M}{Y'}$ as required.

($\Leftarrow$) If  $\marking{N}{X}\dtrans{\alpha}{\gamma} \marking{N}{X'}$ and $\marking{M}{Y}\dtrans{\gamma}{\beta}\marking{M}{Y'}$ for some $\alpha\in\{0,1\}^k$, $\beta\in\{0,1\}^m$, $\gamma\in\{0,1\}^l$ then there exists mutually independent $U\subseteq T_N$ with $\source{U}=\alpha$, $\target{U}=\gamma$, and mutually independent $V\subseteq T_M$ with $\source{V}=\gamma$, $\target{V}=\beta$. In particular, $(U,V)$ is a synchronisation and so, using the conclusion of Lemma~\ref{lem:minimalDecomposition}, there exists a mutually independent family $\{(U_i,V_i)\}_{i\in I}$ of minimal synchronisations with $\bigcup_i U_i=U$ and $\bigcup_i V_i =V$. This family witnesses the transition $\marking{N\comp M}{X\uplus Y}\dtrans{\alpha}{\beta} \marking{N\comp M}{X'\uplus Y'}$. \qed
\end{proof}

\paragraph{Proof of Corollary~\ref{cor:traces}.}
\begin{proof}
Simple induction on $p$, using the conclusion of Theorem~\ref{thm:compositionality}. \qed
\end{proof}

\begin{lemma}\label{lem:decompWeakTrace}
Suppose that $N:k\to l$ and $M:l\to m$ are nets with boundaries. If there is a trace \[ \marking{N\comp M}{X\uplus Y} (\xrightarrow{\epsilon_{k,m}})^* \marking{N\comp M}{X'\uplus Y'}\]
 then there exists $p\in \N$, $\gamma_i\in\{0,1\}^l$ for $1\leq i\leq p$ and traces
\[
\marking{N}{X} \dtrans{0^k}{\gamma_1} \marking{N}{X_1}
\dots \dtrans{0^k}{\gamma_k} \marking{N}{X'}
\]
\[
\marking{M}{Y} \dtrans{\gamma_1}{0^m}\marking{M}{Y_1}
\dots \dtrans{\gamma_k}{0^m}\marking{M}{Y'}.
\]
\end{lemma}
\begin{proof}
Induction on the length of the trace, using the conclusion of Theorem~\ref{thm:compositionality}.  \qed
\end{proof}

\paragraph{Proof of Theorem~\ref{thm:weakCompositionality}.}
\begin{proof}
(i) Suppose that $\marking{N\comp M}{X\uplus Y} \dTrans{\alpha}{\beta} \marking{N\comp M}{X'\uplus Y'}$ for some $\alpha\in\{0,1\}^k$ and $\beta\in\{0,1\}^n$. Then, by definition, there exist $X''\uplus Y''$, $X'''\uplus Y'''$ with
\[
\marking{N\comp M}{X\uplus Y} (\epstrans{k}{m})^*
\marking{N\comp M}{X''\uplus Y''} \dtrans{\alpha}{\beta}
\marking{N\comp M}{X'''\uplus Y'''} (\epstrans{k}{m})^*
\marking{N\comp M}{X'\uplus Y'}
\]
Now we use the conclusions of Lemma~\ref{lem:decompWeakTrace} and Theorem~\ref{thm:compositionality} to obtain the required traces.

(ii) If $\marking{N}{X}\dTrans{\alpha}{\gamma} \marking{N}{X'}$ and $\marking{M}{Y}\dTrans{\gamma}{\beta}\marking{M}{Y'}$ for some $\alpha\in\{0,1\}^k$, $\beta\in\{0,1\}^m$, $\gamma\in\{0,1\}^l$, then there exist $p_N,q_N,p_M,q_M\in\N$,  $X'',X'''\subseteq P_N$, $Y'',Y'''\subseteq P_M$ and traces
\[
\marking{N}{X} (\xrightarrow{\epsilon_{k,l}})^{p_N}
\marking{N}{X''}  \dtrans{\alpha}{\gamma} \marking{N}{X'''}
(\xrightarrow{\epsilon_{k,l}})^{q_N} \marking{N}{X'}
\]
\[
\marking{M}{Y} (\xrightarrow{\epsilon_{l,m}})^{p_M}
\marking{M}{Y''}  \dtrans{\gamma}{\beta} \marking{M}{Y'''}
(\xrightarrow{\epsilon_{l,m}})^{q_M} \marking{M}{Y'}
\]
Now, using the fact that each net in any marking can make $\epsilon$ transition and remain in the same marking (witnessing the firing of the empty set of transitions), we can use Theorem~\ref{thm:compositionality} to obtain:
\begin{multline*}
\marking{N\comp M}{X\uplus Y}
(\xrightarrow{\epsilon_{k,m}})^{p_N}
\marking{N\comp M}{X''\uplus Y}
(\xrightarrow{\epsilon_{k,m}})^{p_M} \\
\marking{N\comp M}{X''\uplus Y''}
\dtrans{\alpha}{\beta}
\marking{N\comp M}{X'''\uplus Y'''} \\
(\xrightarrow{\epsilon_{k,m}})^{q_N}
\marking{N\comp M}{X'\uplus Y'''}
(\xrightarrow{\epsilon_{k,m}})^{q_M}
\marking{N\comp M}{X'\uplus Y'}
\end{multline*}
and thus $\marking{N\comp M}{X\uplus Y}\dTrans{\alpha}{\beta}\marking{N\comp M}{X'\uplus Y'}$ as required. \qed
\end{proof}

\paragraph{Proof of Theorem~\ref{thm:correctness}.}
\begin{proof}
We prove this by structural induction on $t$. The base case, when $t$ is a variable, trivially holds. The interesting inductive case is $t\comp t'$. We must show that
$\epsmin{ \epsminpr{t}_{(\mathcal{V},\mathcal{F})}
       \comp \epsminpr {t'}_{(\mathcal{V},\mathcal{F})} }$ ($\dagger$) is isomorphic to
$\epsmin{ \NFA{ \semanticsOf{t\comp t'}_\mathcal{V} } { \mrk{t\comp t'}{\mathcal{I}} }
                  { \mrk{t\comp t'}{\mathcal{F}} } }$.
Using the definitions of $\semanticsOf{-}_\mathcal{V}$ and
$\mrk{-}{}$:
\begin{multline}\label{eq:simplification}
\epsmin{ \NFA{\semanticsOf{t\comp t'}_\mathcal{V}}
                         {\mrk{t\comp t'}{\mathcal{I}} }
                         {\mrk{t\comp t'}{\mathcal{F}} } }\\
=\epsmin{ \NFA{\semanticsOf{t}_\mathcal{V}\comp \semanticsOf{t'}_\mathcal{V}}
                           {\mrk{t}{\mathcal{I}} \uplus \mrk{t'}{\mathcal{I}}}
                           { \mrk{t}{\mathcal{F}} \uplus \mrk{t'}{\mathcal{F}} }
                }
\end{multline}
The inductive hypothesis gives us that
\begin{equation}\label{eq:ind1}
\epsminpr{t}_{(\mathcal{V},\mathcal{F})}
    \cong
        \epsmin{ \NFA{\semanticsOf{t}_\mathcal{V}}
                                 {\mrk{t}{\mathcal{I}} }
                                 {\mrk{t}{\mathcal{F}} }
                      }
\end{equation}
and
\begin{equation}\label{eq:ind2}
\epsminpr{t'}_{(\mathcal{V},\mathcal{F})}
    \cong
        \epsmin{ \NFA{\semanticsOf{t'}_\mathcal{V}}
                                 {\mrk{t'}{\mathcal{I}} }
                                 {\mrk{t'}{\mathcal{F}} } }
\end{equation}
Substituting~ \eqref{eq:ind1} and~\eqref{eq:ind2} in ($\dagger$), and using~\eqref{eq:simplification}, our task reduces to showing that:
\begin{multline}
\epsmin{
            \epsmin{ \NFA{\semanticsOf{t}_\mathcal{V}}
                                     {\mrk{t}{\mathcal{I}} }
                                     {\mrk{t}{\mathcal{F}} } } \comp
            \epsmin{ \NFA{\semanticsOf{t'}_\mathcal{V}}
                                     {\mrk{t'}{\mathcal{I}} }
                                     {\mrk{t'}{\mathcal{F}} } } }
\\ \cong
\epsmin{
            \NFA{\semanticsOf{t}_\mathcal{V}\comp \semanticsOf{t'}_\mathcal{V}}
                      {\mrk{t}{\mathcal{I}} \uplus \mrk{t'}{\mathcal{I}}}
                      {\mrk{t}{\mathcal{F}} \uplus \mrk{t'}{\mathcal{F}} } }
\end{multline}
To do this, it is sufficient to show that
\begin{equation} \label{eq:withepsmin}
\epsclose{
\epsmin{ \NFA{\semanticsOf{t}_\mathcal{V}}
                         {\mrk{t}{\mathcal{I}} }
                         {\mrk{t}{\mathcal{F}} } } \comp
\epsmin{ \NFA{\semanticsOf{t'}_\mathcal{V}}
                         {\mrk{t'}{\mathcal{I}} }
                         {\mrk{t'}{\mathcal{F}} } } }
\end{equation}
and
\begin{equation}\label{eq:otherone}
\epsclose{
\NFA{ \semanticsOf{t}_\mathcal{V}\comp \semanticsOf{t'}_\mathcal{V}}
         {\mrk{t}{\mathcal{I}} \uplus \mrk{t'}{\mathcal{I}} }
         {\mrk{t}{\mathcal{F}} \uplus \mrk{t'}{\mathcal{F}} }
}
\end{equation}
recognise the same language, where $\epsclose{-}$ means $\epsilon$-closure.
But~\eqref{eq:withepsmin} recognises the same language
as
\begin{equation}\label{eq:withoutepsmin}
\epsclose{
 \NFA{\semanticsOf{t}_\mathcal{V}}
          {\mrk{t}{\mathcal{I}} }
          {\mrk{t}{\mathcal{F}} }  \comp
 \NFA{\semanticsOf{t'}_\mathcal{V}}
           {\mrk{t'}{\mathcal{I}} }
           {\mrk{t'}{\mathcal{F}} } }
\end{equation}
and now the translation between paths in~\eqref{eq:otherone} and~\eqref{eq:withoutepsmin} follows directly from the conclusion of Theorem~\ref{thm:weakCompositionality}. \qed
\end{proof}

\begin{figure}
\[
\includegraphics[height=10cm]{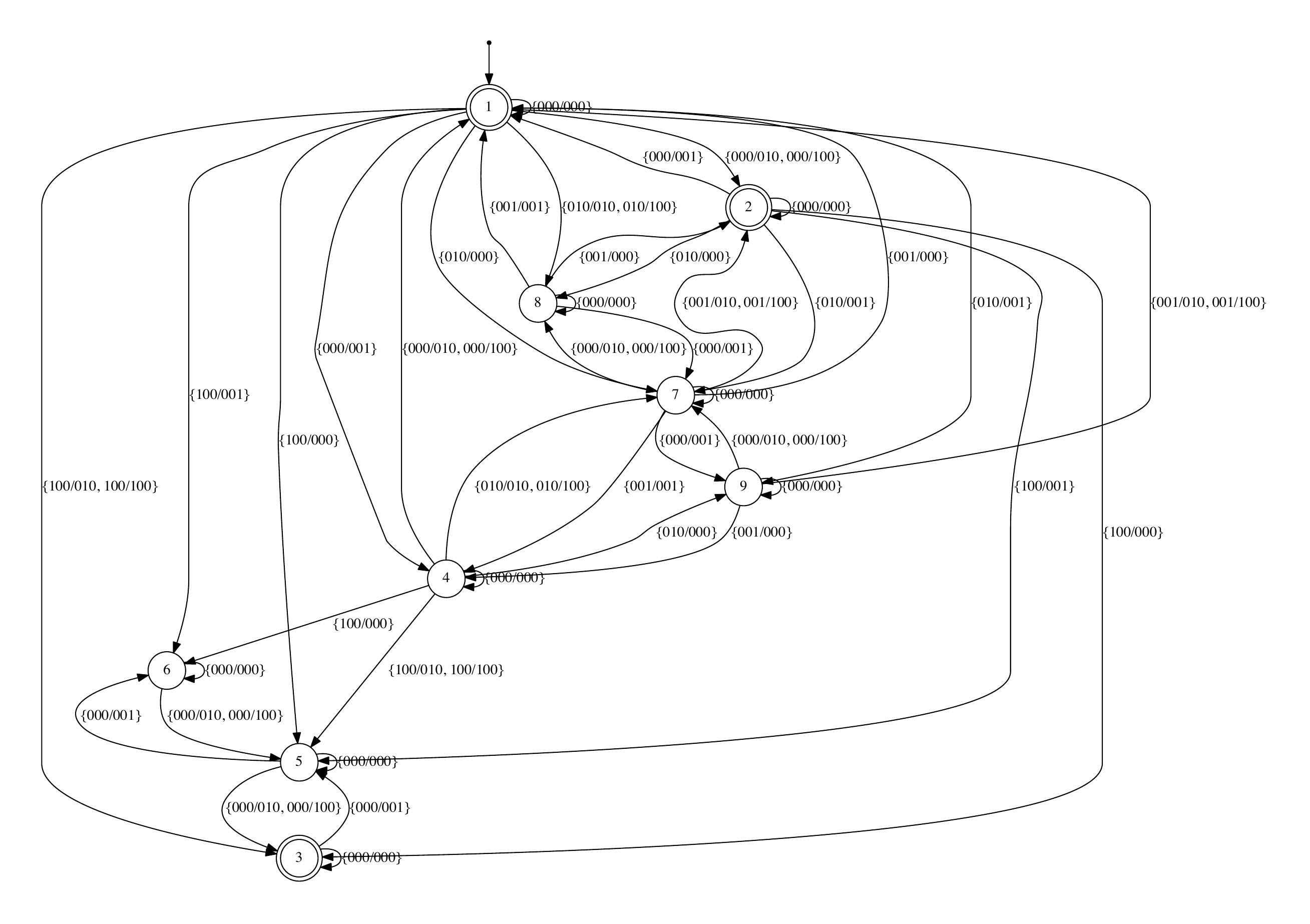}
\]
\caption{Fixed point reached at minimal DFA for $PhRow_3$, error state not drawn.\label{fig:philofixedpoint}}
\end{figure}

\end{document}